\newif\ifpictures
\picturestrue

\documentclass[12pt]{amsart}
\usepackage{amssymb,amsmath}
 \usepackage{amsopn}
 \usepackage{mathabx}
 \usepackage{xspace}
  \usepackage{hyperref}
 \usepackage[dvips]{graphicx}
\usepackage[arrow,matrix,curve]{xy}
\usepackage {color, tikz}
\usepackage{wasysym}

\numberwithin{equation}{section}
\newtheorem{thm}{Theorem}
\newtheorem{prop}[thm]{Proposition}
\newtheorem{lemma}[thm]{Lemma}
\newtheorem{cor}[thm]{Corollary}

\theoremstyle{definition}

\newtheorem{definition}[thm]{Definition}

\numberwithin{thm}{section}

\newcounter{FNC}[page]
\def\newfootnote#1{{\addtocounter{FNC}{2}$^\fnsymbol{FNC}$%
     \let\thefootnote\relax\footnotetext{$^\fnsymbol{FNC}$#1}}}

\newcommand{\N}{\mathbb{N}}

\newcommand{\R}{\mathbb{R}}

\newcommand\cH{{\ensuremath{\mathcal{H}}}\xspace}
\newcommand\cI{{\ensuremath{\mathcal{I}}}\xspace}
\newcommand\cJ{{\ensuremath{\mathcal{J}}}\xspace}

\newcommand\cP{{\ensuremath{\mathcal{P}}}\xspace}

\newcommand\cV{{\ensuremath{\mathcal{V}}}\xspace}

\newcommand{\alp}{\alpha}
\newcommand{\lam}{\lambda}

\newcommand{\Sig}{\Sigma}

\definecolor{DarkGreen}{rgb}{0,0.65,0}

\newcommand{\struc}[1]{{\color{blue} #1}}

\DeclareMathOperator{\conv}{conv}

\DeclareMathOperator{\New}{New}

\DeclareMathOperator{\poly}{poly}

\def\endexa{\hfill$\hexagon$}

\author{Mareike Dressler} 

\address{Mareike Dressler, Goethe-Universit\"at, FB 12 -- Institut f\"ur Mathematik,
	Postfach 11 19 32, 60054 Frankfurt am Main, Germany\medskip}
\email{dressler@math.uni-frankfurt.de}

\author{Adam Kurpisz} 

\address{Adam Kurpisz, ETH Z\"urich --  Department of Mathematics, R\"amistrasse 101,	8092 Z\"urich, Switzerland\medskip}
\email{adam.kurpisz@ifor.math.ethz.ch}

\author{Timo de Wolff}

\address{Timo de Wolff, Technische Universit\"at Berlin, Institut f\"ur Mathematik, Stra{\ss}e des 17.~Juni 136, 10623 Berlin,
 Germany\medskip}

\email{dewolff@math.tu-berlin.de}


\title[Optimization over the Boolean Hypercube via SONCs]{Optimization over the Boolean Hypercube via Sums of Nonnegative Circuit Polynomials}

\begin{document}

\begin{abstract}
Various key problems from theoretical computer science can be expressed as polynomial optimization problems over the boolean hypercube.
One particularly successful way to prove complexity bounds for these types of problems are based on sums of squares (SOS) as nonnegativity certificates. In this article, we initiate the analysis of optimization problems over the boolean hypercube via a recent, alternative certificate called sums of nonnegative circuit polynomials (SONC). We show that key results for SOS based certificates remain valid: First, for polynomials, which are nonnegative over the $n$-variate boolean hypercube with constraints of degree $d$ there exists a SONC certificate of degree at most $n+d$. Second, if there exists a degree $d$ SONC certificate for nonnegativity of a polynomial over the boolean hypercube, then there also exists a short degree $d$ SONC certificate, that includes at most $n^{O(d)}$ nonnegative circuit polynomials.

\end{abstract} 
 
\maketitle

\section{Introduction}

An \struc{\textit{optimization problem over a boolean hypercube}} is an $n$-variate (constrained) polynomial optimization problem where the feasibility set is restricted to the vertices of an $n$-dimensional hypercube
\begin{eqnarray}
\label{Eq:mat_prog_formulation}
 & & \min f(\mathbf{x}), \nonumber \\
 & \text{ subject to} & p_1(\mathbf{x}),\ldots,p_m(\mathbf{x}) \geq 0, \label{Equ:HypercubeOptimizationProblem} \\ 
 & & \mathbf{x} \in \{a_i,b_i\}^n, \nonumber \\ 
 & & f,p_1,\ldots,p_m \in \R[\mathbf{x}]. \nonumber
\end{eqnarray}

The formulation~\eqref{Eq:mat_prog_formulation} captures a class of optimization problems, that belong to the core of theoretical computer science. However, it is known that solving the above formulation is NP-hard in general, since one can easily cast, e.g., the \textsc{Independent Set} problem in this framework.

One of the most promising approaches in constructing efficient algorithms is the \struc{\textit{sum of squares (SOS) hierarchy}}~\cite{GrigorievV01,Nesterov00,parrilo00,schor87}, also known as \struc{\textit{Lasserre relaxation}} \cite{Lasserre:LasserreRelaxation}. The method is based on a Positivstellensatz result~\cite{Putinar:Positivstellensatz} saying that the polynomial $f$, which is nonnegative over the feasibility set given in~\eqref{Eq:mat_prog_formulation}, can be expressed as a sum of squares times the constraints defining the set. Bounding a maximum degree of a polynomial used in a representation of $f$ provides a family of algorithms parametrized by an integer $d$. Finding a degree $d$ SOS certificate for nonnegativity of $f$ can be performed by solving a \struc{\emph{semidefinite programming} (SDP)} formulation of size $n^{O(d)}$. Finally, for every (feasible) $n$-variate unconstrained hypercube optimization problem there exists a degree $2n$ SOS certificate.

On the one hand, the SOS algorithm provide the best available approximation algorithms for
a wide variety of optimization problems.
For example, the degree 2 SOS for the \textsc{Independent Set} problem implies the Lov\'{a}sz $\theta$-function~\cite{Lovasz79} and gives the Goemans-Williamson relaxation~\cite{GoemansW95} for the \textsc{Max Cut} problem. The ARV algorithm of the \textsc{Sparsest Cut} \cite{AroraRV09} problem can be captured by SOS of degree 6. Finally, the subexponential time algorithm for \textsc{Unique Games}~\cite{AroraBS10} is implied by a SOS of sublinear degree~\cite{BarakRS11,GuruswamiS11}. More recently, it has been shown that $O(1)$ degree SOS is equivalent in power to any polynomial size SDP extended formulation in approximating maximum constraint satisfaction problems \cite{LeeRagSteu15}.
Other applications of the SOS method for combinatorial optimization can be found in~\cite{BarakRS11,BateniCG09,Chlamtac07,ChlamtacS08,DBLP:conf/soda/CyganGM13,VegaK07,GuruswamiS11,Mastrolilli17,MagenM09,RaghavendraT12}. For a more detailed overview on the use of SOS in approximation algorithms, see the surveys~\cite{Chla12,Laurent03,laurent09}.

On the other hand, it is known that the SOS algorithm admits certain weaknesses.  First, for some hypercube optimization problems the SOS algorithm performs much worse than other known methods. Grigoriev in~\cite{Grigoriev01} shows that a $\Omega(n)$ degree SOS certificate is needed to detect that the \textsc{Knapsack} instance $\{x\in [0,1]^n : \sum_{i=1}^n x_i= \lfloor n/2 \rfloor+1/2\}$ contains no integer point. Simpler proofs can be found in~\cite{GrigorievHP02,Laurent03a, KurpiszLM16}. Other SOS degree lower bounds for \textsc{Knapsack} problems appeared in~\cite{Cheung07,KurpiszLM17}.
Another example is the problem of scheduling unit size jobs on a single machine to minimize the number of late jobs. The problem is solvable in polynomial time using the Moore-Hodgson algorithm; an $\Omega(\sqrt{n})$ degree SOS algorithm, however, still attains an unbounded integrality gap~\cite{KurpiszLM17b}.
For other important SOS limitations see e.g.~\cite{MekaPW15,BarakHKKMP16}.

Second, it remains open if finding a degree $d$ SOS certificate can be performed in time $n^{O(d)}$. Indeed, as noted in the recent paper by O'Donnell~\cite{ODonnell17} and further discussed by Raghavendra and Weitz in~\cite{RaghavendraW17} it is not obviously true that the search can be done so efficiently.
Namely, even if a small degree SOS certificate exists, the polynomials in the certificate do not have necessarily small coefficients.
O'Donnell in~\cite{ODonnell17} gives an example of a polynomial optimization problem that admits a degree 2 SOS certificate, but every degree 2 SOS certificate for this problem has exponential bit complexity. Moreover, in \cite{RaghavendraW17} the example is modified and cast into a hypercube optimization problem again having a degree 2 SOS certificate, which, however, has super-polynomial bit complexity for certificates up to the degree $O(\sqrt{n})$.
For small $d$, this excludes the possibility that known optimization tools used for solving SDP problems like the ellipsoid method~\cite{KHACHIYAN1980,GroetschelLovaszSchrijver1988} are able to find a degree $d$ certificates in time $n^{O(d)}$ for optimization problems of the form~\eqref{Eq:mat_prog_formulation}.
%
%
The above arguments motivate the search of new methods for solving hypercube optimization problems efficiently.
\medskip

In this article, we initiate an analysis of hypercube optimization problems of the form \eqref{Equ:HypercubeOptimizationProblem} via \struc{\textit{sums of nonnegative circuit polynomials (SONC)}}. SONCs are a nonnegativity certificate introduced recently by Iliman and the third author \cite{Iliman:deWolff:Circuits}, which are independent of sums of squares; see Definition \ref{Def:CircuitPolynomial} and Theorem \ref{Thm:ConeContainment} for further details. Similarly as Lasserre relaxation for SOS, a Schm\"udgen-like Positivstellensatz yields a converging hierarchy of lower bounds for polynomial optimization problems with compact constraint set; see \cite[Theorem 4.8]{Dressler:Iliman:deWolff:Positivstellensatz} and Theorem \ref{thm:SONC_Positivstellensatz}. These bounds can be computed via a convex optimization program called \struc{\textit{relative entropy programming (REP)}} \cite[Theorem 5.3]{Dressler:Iliman:deWolff:Positivstellensatz}.
Our main question in this article is:
\begin{quote}
 \textit{Can SONC certificates be an alternative for SOS methods for optimization problems over the hypercube?}
\end{quote}
We answer this question affirmatively in the sense that we prove SONC complexity bounds for \eqref{Equ:HypercubeOptimizationProblem} analogous to the SOS bounds mentioned above. More specifically, we show:
\begin{enumerate}
 \item For every polynomial which is nonnegative over an $n$-variate hypercube with constraints of degree at most $d$ there exists a SONC certificate of nonnegativity of degree at most $n+d$; see Theorem \ref{thm:f_vanishing_on_H_is_SONC_2n} and Corollary \ref{Cor:DegreeNplusDCertificate}.
 \item If a polynomial $f$ admits a degree $d$ SONC certificate of nonnegativity over an $n$-variate hypercube, then the polynomial $f$ admits also a \emph{short} degree $d$ SONC certificate that includes at most $n^{O(d)}$ nonnegative circuit polynomials; see Theorem \ref{thm:degree_d_certificate}.
\end{enumerate}

For a discussion and remaining open problems, to turn these results into an efficient algorithms, see Section \ref{Subsec:Prelim:Degree_d_SONC_certificates} and the end of Section \ref{sec:degree_d_certificate}. 

\smallskip

Furthermore, we show some structural properties of SONCs:

\begin{enumerate}
	\item We give a simple, constructive example showing that the SONC cone is not closed under multiplication. Subsequently we use this construction to show that the SONC cone is neither closed under taking affine transformations of variables, see Lemma \ref{lem:sonc_not_closed_mulitilication} and Corollary \ref{cor:sonc_not_closed_transform_variables}.
	\item We address an open problem raised in~\cite{Dressler:Iliman:deWolff:Positivstellensatz} asking whether the Schm\"udgen-like Positivstellensatz for SONCs (Theorem \ref{thm:SONC_Positivstellensatz}) can be improved to an equivalent of Putinar's Positivstellensatz \cite{Putinar:Positivstellensatz}. We answer this question negatively by showing an explicit hypercube optimization example, which provably does not admit a Putinar representation for SONCs; see Theorem \ref{thm:counterexample} and the discussion afterwards.
\end{enumerate}

Our article is organized as follows: In Section \ref{Sec:Preliminaries} we introduce the necessary background about SONCs. In Section \ref{Sec:PropertiesSONCCone} we show that the SONC cone is closed neither under multiplication nor under affine transformations. In Section \ref{sec:sonc_hypercube} we provide our two main results regarding the degree bounds for SONC certificates over the hypercube. In Section \ref{Sec:PutinarPositivstellensatz} we prove the non-existence of an equivalent of Putinar's Positivstellensatz for SONCs and discuss this result.

\subsection*{Acknowledgements}
AK was supported by the Swiss National Science Foundation project PZ00P2$\_$174117
``Theory and Applications of Linear and Semidefinite Relaxations for Combinatorial Optimization Problems''.
TdW was supported by the DFG  grant WO 2206/1-1. This article was finalized while TdW was hosted by the Institut Mittag-Leffler. We thank the institute for its hospitality.

\section{Preliminaries}
\label{Sec:Preliminaries}

In this section we collect basic notions and statements on sums of nonnegative circuit polynomials (SONC).\\
Throughout the paper, we use bold letters for vectors, e.g., $\struc{\mathbf{x}}=(x_1,\ldots,x_n) \in \R^n$. Let $\struc{\N^*} = \N \setminus \{\mathbf{0}\}$ and $\struc{\R_{\geq 0}}$ ($\struc{\R_{> 0}}$) be the set of nonnegative (positive) real numbers. Furthermore let \struc{$\R[\mathbf{x}] = \R[x_1,\ldots,x_n]$} be the ring of real $n$-variate polynomials and the set of all $n$-variate polynomials of degree less than or equal to $2d$ is denoted by $\struc{\R[\mathbf{x}]_{n,2d}}$. We denote by $\struc{[n]}$ the set $\{1,\ldots,n\}$ and the sum of binomial coefficients $\sum_{k=0}^d \binom{n}{k}$ is abbreviated by $\struc{\binom{n}{\leq d}}$. 
Let $\struc{\mathbf{e_1}},\ldots,\struc{\mathbf{e_n}}$ denote the canonical basis vectors in $\R^n$. 

%
%

\subsection{Sums of Nonnegative Circuit Polynomials}

Let $\struc{A} \subset \N^n$ be a finite set.
In what follows, we consider polynomials $f \in \R[\mathbf{x}]$ supported on $A$. Thus, $f$ is of the form $\struc{f(\mathbf{x})} = \sum_{\boldsymbol{\alp} \in A}^{} f_{\boldsymbol{\alp}}\mathbf{x}^{\boldsymbol{\alp}}$ with $\struc{f_{\boldsymbol{\alp}}} \in \R$ and $\struc{\mathbf{x}^{\boldsymbol{\alp}}} = x_1^{\alp_1} \cdots x_n^{\alp_n}$. A lattice point is called \struc{\textit{even}} if it is in $(2\N)^n$ and a term $ f_{\boldsymbol{\alp}}\mathbf{x}^{\boldsymbol{\alp}}$ is called a \struc{\emph{monomial square}} if $f_{\boldsymbol{\alp}} > 0$ and $\boldsymbol{\alp}$ even. We denote by $\struc{\New(f)} = \conv\{\boldsymbol{\alp} \in \N^n : f_{\boldsymbol{\alp}} \neq 0\}$ the Newton polytope of $f$. \\

Initially, we introduce the foundation of SONC polynomials, namely \textit{circuit polynomials}; see also \cite{Iliman:deWolff:Circuits}:

\begin{definition}
	A polynomial $f \in \R[\mathbf{x}]$ 
	is called a \struc{\emph{circuit polynomial}} if it is of the form
	\begin{eqnarray}
	\struc{f(\mathbf{x})} & := & \sum_{j=0}^r f_{\boldsymbol{\alp}(j)} \mathbf{x}^{\boldsymbol{\alp}(j)} + f_{\boldsymbol{\beta}} \mathbf{x}^{\boldsymbol{\beta}}, \label{Equ:CircuitPolynomial}
	\end{eqnarray}
	with $\struc{r} \leq n$, exponents $\struc{\boldsymbol{\alp}(j)}$, $\struc{\boldsymbol{\beta}} \in A$, and coefficients $\struc{f_{\boldsymbol{\alp}(j)}} \in \R_{> 0}$, $\struc{f_{\boldsymbol{\beta}}} \in \R$, such that the following conditions hold:
	
	\begin{description}
		\item[(C1)] $\New(f)$ is a simplex with even vertices $\boldsymbol{\alp}(0), \boldsymbol{\alp}(1),\ldots,\boldsymbol{\alp}(r)$.
		\item[(C2)] 
		The exponent $\boldsymbol{\beta}$ is in the strict interior of $\New(f)$. Hence, there exist unique \struc{\emph{barycentric coordinates} $\lambda_j$} relative to the vertices $\boldsymbol{\alp}(j)$ with $j=0,\ldots,r$ satisfying
		\begin{eqnarray*}
			& & \boldsymbol{\beta} \ = \ \sum_{j=0}^r \lambda_j \boldsymbol{\alp}(j) \ \text{ with } \ \lambda_j \ > \ 0 \ \text{ and } \  \sum_{j=0}^r \lambda_j \ = \ 1.
		\end{eqnarray*}
	\end{description}
	We call the terms $f_{\boldsymbol{\alp}(0)} \mathbf{x}^{\boldsymbol{\alp}(0)},\ldots,f_{\boldsymbol{\alp}(r)} \mathbf{x}^{\boldsymbol{\alp}(r)}$ the \struc{\emph{outer terms}} and $f_{\boldsymbol{\beta}} \mathbf{x}^{\boldsymbol{\beta}}$ the \struc{\emph{inner term}} of $f$. 
	
	For every circuit polynomial we define the corresponding \struc{\textit{circuit number}} as
	\begin{eqnarray}
	\struc{\Theta_f} \ := \ \prod_{j = 0}^r \left(\frac{f_{\boldsymbol{\alp}(j)}}{\lambda_j}\right)^{\lambda_j}. \label{Equ:DefCircuitNumber}
	\end{eqnarray}
	\label{Def:CircuitPolynomial}
	\endexa
\end{definition}


The first fundamental statement about circuit polynomials is that its nonnegativity is determined by its circuit number $\Theta_f$ and $f_{\boldsymbol{\beta}}$ entirely:

\begin{thm}[\cite{Iliman:deWolff:Circuits}, Theorem 3.8]
	Let $f$  be a  circuit polynomial with inner term $f_{\boldsymbol{\beta}} \mathbf{x}^{\boldsymbol{\beta}}$ and let $\Theta_f$ be the corresponding circuit number, as defined in \eqref{Equ:DefCircuitNumber}. Then the following statements are equivalent:
	\begin{enumerate}
		\item $f$ is nonnegative.
		\item $|f_{\boldsymbol{\beta}}| \leq \Theta_f$ and $\boldsymbol{\beta} \not \in (2\N)^n$ \quad or \quad $f_{\boldsymbol{\beta}} \geq -\Theta_f$ and $\boldsymbol{\beta }\in (2\N)^n$.
	\end{enumerate}
	\label{Thm:CircuitPolynomialNonnegativity}
\end{thm}


Therefore, expressing a polynomial as a \struc{\it{sum of nonnegative circuit polynomials (SONC)}} is a certificate for the polynomials nonnegativity.

\begin{definition}
	We define for every $n,d \in \N^*$ the set of \struc{\emph{sums of nonnegative circuit polynomials} (SONC)} in $n$ variables of degree $2d$ as
	$$\struc{C_{n,2d}} \ := \ \left\{f \in \R[\mathbf{x}]_{n,2d} \ :\  f = \sum_{i=1}^k \mu_i p_i, \begin{array}{c}
	p_i \text{ is a nonnegative circuit polynomial, } \\
	\mu_i \geq 0, k \in \N^* \\
	\end{array}
	\right\}
	$$
	\label{Def:SONC}
	\endexa
\end{definition}

We denote by SONC both the set of SONC polynomials and the property of a polynomial to be a sum of nonnegative circuit polynomials.

In what follows let $\struc{P_{n,2d}}$ be the cone of nonnegative  $n$-variate polynomials of degree at most $2d$ and $\struc{\Sig_{n,2d}}$ be the corresponding cone of sums of squares respectively. 
An important observation is, that SONC polynomials form a convex cone independent of the SOS cone:

\begin{thm}[\cite{Iliman:deWolff:Circuits}, Proposition 7.2]
	$C_{n,2d}$ is a convex cone satisfying:
	\begin{enumerate}
		\item $C_{n,2d} \subseteq P_{n,2d}$ for all $n,d \in \N^*$,
		\item $C_{n,2d} \subseteq \Sigma_{n,2d}$ if and only if $(n,2d)\in\{(1,2d),(n,2),(2,4)\}$,
		\item  $\Sigma_{n,2d} \not\subseteq C_{n,2d}$ for all $(n,2d)$ with $2d \geq 6$.
	\end{enumerate}
	\label{Thm:ConeContainment}
\end{thm}

For further details about the SONC cone see \cite{deWolff:Circuits:OWR,Iliman:deWolff:Circuits, Dressler:Iliman:deWolff:Positivstellensatz}.
 
\subsection{SONC Certificates over a Constrained Set}
In \cite[Theorem 4.8]{Dressler:Iliman:deWolff:Positivstellensatz}, Iliman, the first, and the third author showed that for an arbitrary real polynomial which is strictly positive on a compact, basic closed semialgebraic set $K$ there exists a SONC certificate of nonnegativity. Hereinafter we recall this result.

\medskip

We assume that $K$ is given by polynomial inequalities $\struc{g_i(\mathbf{x})} \geq 0$ for $i = 1,\ldots,s$ and is compact. For technical reason we add $2n$ redundant box constraints $\struc{l_j(\mathbf{x})} := N\pm x_j\geq 0$ for some sufficiently large $N \in \N$, which always exists due to our assumption of compactness of $K$; see \cite{Dressler:Iliman:deWolff:Positivstellensatz} for further details. Hence, we have

\begin{align}
\struc{K}\ :=\ \{\mathbf{x}\in \R^n: g_i(\mathbf{x})\geq 0 \text{ for } i\in [s] \text{ and } l_j(\mathbf{x})\geq 0 \text{ for } j\in [2n]\}.
\label{Equ:ConstraintKbox}
\end{align}
In what follows we consider polynomials $H^{(q)}(\mathbf{x})$ defined as products of at most $q \in \N^*$ of the polynomials $g_i,l_j$ and $1$, i.e.,
\begin{align}
\struc{H^{(q)}(\mathbf{x})} \ := \ \prod_{k=1}^{q} h_k(\mathbf{x}),
\label{Equ:ProductConstraint}
\end{align}
where $\struc{h_k}\in \{1,g_1,\ldots,g_s,l_1,\ldots,l_{2n}\}$. 
Now we can state:
\begin{thm}
Let $f,g_1,\ldots,g_s\in \R[\mathbf{x}]$ be real polynomials and $K$ be a compact, basic closed semialgebraic set as in \eqref{Equ:ConstraintKbox}. If $f > 0 $ on $K$ then there exist $d,q \in \N^*$ such that we have an explicit representation of $f$ of the following form:
\[
f(\mathbf{x}) \ = \ \sum_{\text{\rm finite }} s(\mathbf{x})H^{(q)}(\mathbf{x}), 
\]
where the $s(\mathbf{x})$ are contained in $C_{n,2d}$ and every $H^{(q)}(\mathbf{x})$ is a product as in \eqref{Equ:ProductConstraint}.
\label{thm:SONC_Positivstellensatz}
\end{thm}

The central object of interest is the smallest value of $d$ and $q$ that allows $f$ a decomposition as in Theorem~\ref{thm:SONC_Positivstellensatz}. This motivates the following definition of a \emph{degree $d$ SONC certificate}.

\begin{definition}
	\label{Def:degree_d_certificate}
	Let $f \in \R[\mathbf{x}]$ such that $f$ is positive on the set $K$ given in \eqref{Equ:ConstraintKbox}. Then $f$ has a \struc{\emph{degree $d$ SONC certificate}} if it admits for some $q \in \N^*$ the following decomposition: 
	$$
	f(\mathbf{x}) \ = \ \sum_{\text{\rm finite }} s(\mathbf{x})H^{(q)}(\mathbf{x}), 
	$$
	where the $s(\mathbf{x})$ are SONCs, the $H^{(q)}(\mathbf{x})$ are products as in \eqref{Equ:ProductConstraint}, and 
	$$\deg\left(\sum_{\text{\rm finite }} s(\mathbf{x})H^{(q)}(\mathbf{x})\right) \ \leq \ d.$$
	\endexa
\end{definition}

\subsection{The Complexity of Finding a Degree $d$ SONC Certificate}
\label{Subsec:Prelim:Degree_d_SONC_certificates}
One can decide nonnegativity for a single given circuit polynomial by solving a system of linear equations. This is due to Theorem \ref{Thm:CircuitPolynomialNonnegativity} and the fact that the $\lam_j$ are unique (and thus trivially nonnegative), since the $\boldsymbol{\alp}(j)$ are affinely independent by construction.

For finding a degree $d$ SONC certificate, there are two main bottlenecks that might effect its complexity. The first one is to guarantee the existence of a sufficiently short degree $d$ SONC certificate. If the first bottleneck is resolved, then a second one might occur: even if the existence of a short degree $d$ SONC certificate is guaranteed, then it is not clear a priori, whether one can search through the space of $n$-variate circuit polynomials of degree at most $d$ efficiently, in order to find such a short certificate.

\medskip

Regarding the first bottleneck, given the fact that a polynomial $f$ admits a degree $d$ SONC certificate, it is open whether there also exists a degree $d$ certificate which consists of a bounded (ideally $n^{O(d)}$) number of components. 

The answer to the equivalent question for the SOS degree $d$ certificates follows from the fact that a polynomial is SOS if and only if the corresponding matrix of coefficients of size $n^{O(d)}$, called the \textit{\struc{Gram matrix}}, is positive semidefinite. Since every real, symmetric matrix $M$ that is positive semidefinite admits a decomposition $M=VV^\top$, this yields an explicit SOS certificate including at most $n^{O(d)}$ polynomials squared. For more details we refer the reader to the excellent lecture notes in~\cite{BarakS16}.  

In this paper we resolve the first bottleneck regarding the existence of short SONC certificates affirmatively. Namely, we show that one can always restrict oneself to SONC certificates including at most $n^{O(d)}$ nonnegative circuit polynomials, see Section~\ref{sec:degree_d_certificate} for further details.

\section{Properties of the SONC cone}
\label{Sec:PropertiesSONCCone}

In this section we show that the SONC cone is neither closed under multiplication nor under affine transformations. First, we give a constructive proof for the fact that the SONC cone is not closed under multiplication, which is simpler than the initial proof of this fact in \cite[Lemma 4.1]{Dressler:Iliman:deWolff:Positivstellensatz}. Second, we use our construction to show that the SONC cone is not closed under affine transformation of variables.


\begin{lemma}
	\label{lem:sonc_not_closed_mulitilication}
For every $d \geq 2$, $n \in \mathbb{N^*}$ the SONC cone $C_{n,d}$ is not closed under multiplication in the following sense: if $p_1,p_2 \in C_{n,d}$, then $p_1 \cdot p_2 \not\in C_{n,2d}$ in general.	
\end{lemma}

\begin{proof}
	For every $d=2n$, $n \in \mathbb{N^*}$ we construct two SONC polynomials $p_1$, $p_2 \in C_{n,d}$ such that the product $p_1 p_2$ is an $n$ variate, degree $2d$ polynomial that is not inside $C_{n,2d}$.
	
	Let $n = 2$. We construct the following two polynomials $p_1,~p_2 \in\mathbb{R}[x_1,x_2]$:
	$$
	\struc{p_1(x_1,x_2)}\ :=\ (1-x_1)^2, \qquad \struc{p_2(x_1,x_2)}\ :=\ (1-x_{2})^2.
	$$
	First, observe that $p_1, p_2$ are nonnegative circuit polynomials, since, in both cases, $\lambda_1=\lambda_2=1/2$, $f_{\boldsymbol{\alp}(1)}=f_{\boldsymbol{\alp}(2)}=1$, and $f_{\boldsymbol{\beta}}=-2$, thus $2=\Theta_f \geq |f_{\boldsymbol{\beta}}|$.
	
	Now consider the polynomial $r(x_1,x_2)= p_1 p_2 = \left( (1-x_1) (1-x_{2})  \right)^2$. We show that this polynomial, even though it is nonnegative, is not a SONC polynomial. Note that $r(x_1,x_2)=  1-2x_1-2x_2+4x_1x_2 +x_{1}^2  +x_2^2 -2x_1^2x_2 -2x_1x_2^2 +x_1^2x_2^2$; the support of $r$ is shown in Figure~\ref{pic:pic1}.
	\begin{figure}
		\includegraphics[width=6cm]{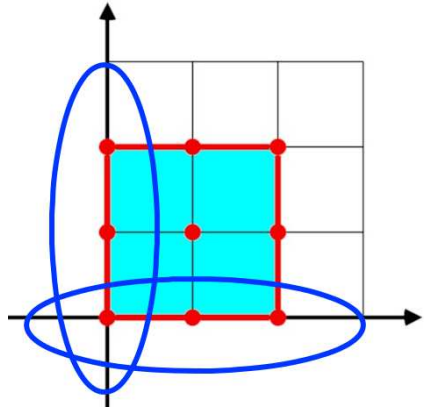}
		\caption{The Newton polytope and the support set of $r(x_1,x_k)$ with the supports of $p_1$ and $p_2$ in blue ovals.}
		\label{pic:pic1}
	\end{figure}
	Assume that $r \in C_{2,4}$, i.e., $r$ has a SONC decomposition. This implies that the term $-2x_1$ has to be an inner term of some nonnegative circuit polynomial $r_1$ in this representation. Such a circuit polynomial necessarily has the terms $1$ and $x_1^2$ as outer terms, that is, 
	$$
	r_1(x_1)\ = \ p_1(x_1,x_2)  \ = \ 1+x_1^2-2x_1
	$$
	Since $\Theta_{r_1}=2$ the polynomial $r_1$ is indeed nonnegative and, in addition, we cannot choose a smaller constant term to construct $r_1$.
	Next, also the term $-2x_2$ has to be an inner term of some nonnegative circuit polynomial $r_2$. Since this term again is on the boundary of $\New(r)$ the only option for such an $r_2$ is: $r_2(x_2)= p_2(x_1,x_2) = 1+x_2^2-2x_2$. However, the term $1$ has been already used in the above polynomial $r_1$, which leads to a contradiction, i.e., $r\notin C_{2,4}$. Since $C_{n,2d} \subseteq C_{n+1,2d}$, the general statement follows.
\end{proof}

Hereinafter we show another operation, which behaves differently for SONC than it does for SOS: Similarly as in the case of multiplications, affine transformations also do not preserve the SONC structure. This observation is important for possible degree bounds on SONC certificates, when considering optimization problems over distinct descriptions of the hypercube.
\begin{cor}
	\label{cor:sonc_not_closed_transform_variables}
	For every $d \geq 4$, $n \in \mathbb{N^*}$ the SONC cone $C_{n,d}$ is not closed under affine transformation of variables.	
\end{cor}
\begin{proof}
	Consider the polynomial $f(x_1,x_2) = x_1^2x_2^2$. Clearly, the polynomial $f$ is a nonnegative circuit polynomial since it is a monomial square, hence $f \in C_{n,d}$. Now consider the following affine transformation of the variables $x_1$ and $x_{2}$:
	$$
	x_1 \rightarrow 1-x_1, \qquad	x_{2} \rightarrow 1-x_{2}.
	$$
	After applying the transformation the polynomial $f$ equals the polynomial $p_1p_2$ from the proof of Lemma~\ref{lem:sonc_not_closed_mulitilication} and thus is not inside $C_{n,d}$.
	
\end{proof}


\section{An Upper Bound on the Degree of SONC Certificates over the Hypercube}
\label{sec:sonc_hypercube}
 In the previous section we showed that the SONC cone is not closed under taking an affine transformation of variables, Corollary~\ref{cor:sonc_not_closed_transform_variables}. Thus, if a polynomial $f$ admits a degree $d$ SONC certificate proving that it is nonnegative on a given compact semialgebraic set $K$, then it is a priori not clear whether a polynomial $g$, obtained from $f$ via an affine transformation of variables, admits a degree $d$ SONC certificate of nonnegativity on $K$, too. The degree needed to prove nonnegativity of $g$ might be much larger than $d$ according to the argumentation in the proof of Corollary~\ref{cor:sonc_not_closed_transform_variables}.
 
 
 In this section we prove that every $n$-variate polynomial which is nonnegative over the boolean hypercube has a degree $n$ SONC certificate. Moreover, if the hypercube is additionally constrained with some polynomials of degree at most $d$, then the nonnegative polynomial over such a set has degree $n+d$ SONC certificate. We show this fact for all hypercubes $\{a_i,b_i\}^n$; see Theorem \ref{thm:upper_bound_SONC_hypercube} for further details.
 
 Formally, we consider the following setting: We investigate real multivariate polynomials in $\R[\mathbf{x}]$. For $j \in [n]$, and $a_j, b_j \in \mathbb{R}$, such that $a_j < b_j$ let 
 $$\struc{g_j(\mathbf{x})} \ := \ (x_j-a_j)(x_j-b_j)$$
 be a \struc{\textit{quadratic polynomial with two distinct real roots}}. Let $\struc{\cH} \subset \R^n$ denote the \struc{$n$-\textit{dimensional hypercube}} given by \struc{$\prod_{j=1}^n\{a_j,  b_j\}$}. 
 Moreover, let 
 $$\struc{\mathcal{P}} \ := \ \{p_1,\ldots,p_m: ~p_i \in \R[\mathbf{x}], ~ i \in [m]\}$$ 
 be a set of polynomials, which we consider as constraints $\struc{p_i(\mathbf{x})} \geq 0$ with $\deg(p_i(\mathbf{x})) \leq d$ for all $i \in [n]$ as follows. We define
 $$\struc{\cH_\mathcal{P}} \ := \ \{\mathbf{x} \in \R^n: ~ g_j(\mathbf{x}) =0,~ j\in [n],~ p(\mathbf{x}) \geq 0,~ p \in \mathcal{P} \}$$
 as the \struc{$n$-\textit{dimensional hypercube} $\cH$ \textit{constrained by polynomial inequalities given by} $\mathcal{P}$}.
 
 Throughout the paper we assume that $|\cP|=\poly(n)$, i.e. the size of the constraint set $\cP$ is polynomial in $n$. This is usually the case, since otherwise the problem gets less tractable from the optimization point of view.

 \medskip
 
 As a first step, we introduce a \textit{Kronecker function}:
 \begin{definition}
 	\label{def:delta_v}
For every $\mathbf{v} \in \cH$ the function
\begin{eqnarray}
\struc{\delta_\mathbf{v} (\mathbf{x})} \ := \ \prod_{j \in [n]:~v_j=a_j} \left( \frac{-x_j +b_j }{b_j-a_j}    \right) \cdot
\prod_{j \in [n]:~v_j=b_j} \left( \frac{x_j - a_j }{b_j-a_j}   \right) \label{Equ:DefKroneckerDelta}
\end{eqnarray}
is called the \struc{\emph{Kronecker delta (function)}} of the vector $\mathbf{v}$.
 \end{definition}
 
Next we show that the term ``Kronecker delta'' is justified, i.e., we show that for every $\mathbf{v} \in \cH$ the function $\delta_{\mathbf{v}}(\mathbf{x})$ takes the value zero for all $\mathbf{x} \in \cH$ except for $\mathbf{x}=\mathbf{v}$ where it takes the value one.
\begin{lemma}
	\label{lem:delta_v_values}
	For every  $\mathbf{v} \in \cH$ it holds that:
	\begin{align*}
	\delta_{\mathbf{v}}(\mathbf{x}) \ = \
	\begin{cases}
	 0,  &\qquad \text{for every } \mathbf{x} \in \cH \setminus \{\mathbf{v}\},\\
	 1,  & \qquad \text{for } \mathbf{x} =\mathbf{v}.
	\end{cases}
	\end{align*}
\end{lemma}
\begin{proof}
	On the one hand, if $\mathbf{x} \in \cH \setminus \{\mathbf{v}\}$, then there exists an index $k$ such that $\mathbf{x}_k\neq \mathbf{v}_k$. This implies that there exists at least one multiplicative factor in $\delta_\mathbf{v}$ which attains the value zero due to \eqref{Equ:DefKroneckerDelta}. On the other hand if $\mathbf{x} =\mathbf{v}$ then we have 
	$$\delta_\mathbf{v} (\mathbf{x}) \ = \ \prod_{j\in [n]:~v_j=a_j} \left( \frac{-a_j +b_j }{b_j-a_j}    \right)
	\prod_{j\in [n]:~v_j=b_j} \left( \frac{b_j - a_j }{b_j-a_j}   \right) \ = \ 1.$$
\end{proof}
 
 The main result of this section is the following theorem.
 
%

 \begin{thm}
	\label{thm:upper_bound_SONC_hypercube}
	Let $f(\mathbf{x}) \in \R[\mathbf{x}]_{n,n}$. Then $f(\mathbf{x}) \geq 0$~ for every $\mathbf{x} \in \cH_{\mathcal{P}}$ if and only if $f$ has the following representation:
{\small	\begin{eqnarray}
\hspace*{1cm} f(\mathbf{x}) & = & \sum_{\mathbf{v} \in \cH_{\mathcal{P}}}  c_\mathbf{v} \delta_\mathbf{\mathbf{v}}(\mathbf{x})   +    \sum_{\mathbf{v} \in \cH \setminus \cH_{\mathcal{P}}}  c_\mathbf{v}  \delta_\mathbf{\mathbf{v}}(\mathbf{x}) p_\mathbf{v}(\mathbf{x})    +    \sum_{j = 1}^n s_j(\mathbf{x}) g_j(\mathbf{x}) + \sum_{j = 1}^n s_{n + j}(\mathbf{x}) (-g_j(\mathbf{x})),\label{Equ:PolynomialNonnegativityCertificate}
\end{eqnarray}}
	where $s_1,\ldots, s_{2n} \in C_{n,n-2}$, $c_\mathbf{v} \in \mathbb{R}_{\geq 0}$ and $p_\mathbf{v} \in \mathcal{P}$. 
\end{thm}
Since we are interested in optimization on the boolean hypercube $\cH$, we assume without loss of generality that the polynomial $f$ considered in Theorem~\ref{thm:upper_bound_SONC_hypercube} has degree at most $n$. Indeed, it $f$ has degree bigger than $n$, one can efficiently reduce the degree of $f$ by applying iteratively the polynomial division with respect to polynomials $g_j$ with $j \in [n]$. The remainder of the division process is a polynomial with degree at most $n$ that agrees with $f$ on all the vertices of $\cH$.

We begin with proving the easy direction of the equivalence stated in Theorem \ref{thm:upper_bound_SONC_hypercube}.

\begin{lemma}
 If $f$ admits a decomposition of the form \eqref{Equ:PolynomialNonnegativityCertificate}, then $f(\mathbf{x})$ is nonnegative for all $\mathbf{x} \in \cH_{\cP}$.
 \label{Lemma:CertificateIsIndeedCertificate}
\end{lemma}

\begin{proof}
The coefficients $c_\mathbf{v}$ are nonnegative, all $s_j(\mathbf{x})$ are SONC and hence nonnegative on $\R^n$. We have $\pm g_j(\mathbf{x}) \geq 0$ for all $\mathbf{x} \in \cH$, and for all choices of $\mathbf{v} \in \cH$ we have $p_{\mathbf{v}}(\mathbf{x}) \geq 0$ for all $\mathbf{x} \in \cH_\cP$, and $\delta_{\mathbf{v}}(\mathbf{x}) \in \{0,1\}$ for all $\mathbf{x} \in \cH$. Thus, the right hand side of \eqref{Equ:PolynomialNonnegativityCertificate} is a sum of positive terms for all $\mathbf{x} \in \cH_\cP$.
\end{proof}

We postpone the rest of the proof of Theorem~\ref{thm:upper_bound_SONC_hypercube} to the end of the section. Now, we state an result about the presentation of the Kronecker delta function $\delta_\mathbf{\mathbf{v}}$. 

\begin{lemma}
	\label{lem:Kronecker_in_SONC_2n}
For every $\mathbf{v} \in \cH$ the Kronecker delta function can be written as
$$\delta_\mathbf{\mathbf{v}} \ = \ \sum_{j=1}^{2^n}s_j H_j^{(n)},$$
for  $s_1,\ldots,s_{2^n} \in \mathbb{R}_{\geq 0}$ and every $H_j^{(n)}$ given as in \eqref{Equ:ProductConstraint} with $q=n$ and $K$ given by the hypercube constraints $g_1,\ldots,g_n$ and $l_1,\ldots,l_{2n}$.
\end{lemma}
\begin{proof}
First note that the function $\delta_\mathbf{v}$ can be rewritten as	
$$
\delta_\mathbf{v} (\mathbf{x})\ = \ \prod_{j=1}^{2n} \frac{1 }{b_j-a_j}    
\prod_{j\in [n]:~v_j=a_j} \left( -x_j +b_j     \right)
\prod_{j\in [n]:~v_j=b_j} \left( x_j - a_j   \right),
$$	
where $\prod_{j=1}^n \frac{1 }{b_j-a_j}  \in \mathbb{R}_{\geq 0}$. Now, the proof follows just by noting that for every $j \in [n]$ both inequalities $-x_j +b_j \geq 0 $ and $x_j - a_j \geq 0$ are in $K$.
\end{proof}

The following statement is well-known in similar variations; see e.g. \cite[Lemma 2.2 and its proof]{BarakS14}. For clarity, we provide an own proof here.
 
\begin{prop}
Let $f \in \R[\mathbf{x}]_{n,2d}$ be a polynomial vanishing on $\cH$. Then $f = \sum_{j = 1}^n p_j g_j$ for some polynomials $p_j \in \R[\mathbf{x}]_{n,2d-2}$.
\label{Prop:HypercubeIdealRepresentation}
\end{prop}

\smallskip

\begin{proof}
	Let $\struc{\cJ} := \langle g_1,\ldots,g_n \rangle$ be the ideal generated by the $g_j$'s. Let $\struc{\cV(\cJ)}$ denote the affine variety corresponding to $\cJ$, $\struc{\cI(\cV(\cJ))}$ denote its radical ideal, and let $\struc{\cI(\cH)}$ denote the ideal of $\cH$. It follows from $\prod_{j= 1}^n g_j \in \cJ$ that $\cV(\cJ) \subseteq \cH$ and hence $\cI(\cH) \subseteq \cI(\cV(\cJ)) = \cJ$. The last equality holds since $\cJ$ itself is a radical ideal. This results from Seidenberg's Lemma; see \cite[Proposition 3.7.15]{Kreuzer:Robbiano} by means of the following observations. The affine variety $\cV(\cJ)$ consists exactly of the points defining $\cH$, therefore we know that $\cJ$ is a zero-dimensional ideal. Furthermore, for every $j\in [n]$ the polynomials $g_j$ satisfy $g_j \in \cJ \cap \R[x_j]$ and $\gcd(g_j,g_j')=1$. Thus, every $f \in \cI(\cH)$ is of the form $f = \sum_{j = 1}^n p_j g_j$.
	
	Moreover $\struc{G} := \{g_1,\ldots,g_n\}$ is a Gr\"obner basis for $\cJ$ with respect to the graded lexicographic order \struc{$\prec_{\rm glex}$}. This follows from Buchberger's Criterion, which says that $G$ is a Gr\"obner basis for $\cJ$ if and only if for all pairs $i\neq j$ the remainder on the division of the $S$-polynomials $\struc{S(g_i,g_j)}$ by $G$ with respect to $\prec_{\rm glex}$ is zero. Consider an arbitrary pair $g_i,g_j$ with $i > j$. Then the corresponding $S$-polynomial is given by 
	\[
	S(g_i,g_j) \ = \ (a_j+b_j) x_i^2 x_j -(a_i+b_i) x_i x_j^2 - a_jb_j x_i^2  + a_ib_i x_j^2  \;.
	\]
	Applying polynomial division with respect to $\prec_{\rm glex}$ yields the remainder $0$ and hence $G$ is a Gr\"obner basis for $\cJ$ with respect to $\prec_{\rm glex}$. Therefore, we conclude that if $f \in \R[\mathbf{x}]_{n,2d}$, then $\deg(p_j)\leq 2d-2$.
\end{proof}

For an introduction to Gr\"obner bases see for example \cite{Cox:Little:OShea}.

\medskip

\begin{thm}
	\label{thm:f_vanishing_on_H_is_SONC_2n}
Let $d \in \N$ and $f \in \R[\mathbf{x}]_{n,2d+2}$ such that $f$ vanishes on $\cH$. Then there exist $s_1,\ldots,s_{2n}$ $\in C_{n,2d}$ such that $f = \sum_{j = 1}^n s_j g_j + \sum_{j = 1}^n s_{n + j} (-g_j)$. 
\end{thm}

\begin{proof}
By Proposition \ref{Prop:HypercubeIdealRepresentation} we know that $f = \sum_{j = 1}^n p_j g_j$ for some polynomials $p_j$ of degree $\leq 2d$. Hence, it is sufficient to show that every single term $p_j g_j$ is of the form $\sum_{j = 1}^n s_j g_j - \sum_{j = 1}^n s_{n + j} g_j$ for some $s_1,\ldots,s_{2n} \in C_{n,2d}$. Let $p_j = \sum_{i = 1}^\ell a_{ji} m_{ji}$ where every $a_{ji} \in \R$ and every $m_{ji}$ is a single monomial. We show that $p_j g_j$  has the desired form by investigating an arbitrary individual term $a_{ji} m_{ji} g_j$.

\textbf{Case 1:} Assume the exponent of $m_{ji}$ is contained in $(2\N)^n$. If $a_{ji} m_{ji}$ is a monomial square, then $a_{ji} m_{ji}$ is a circuit polynomial. If $a_{ji} < 0$, then $-a_{ji} m_{ji}$ is a monomial square. In both cases we obtain a representation $s_{ji} (\pm g_{ji})$, where $s_{ji} \in C_{n,2d}$.

\textbf{Case 2:} Assume the the exponent $\boldsymbol{\beta}$ of $m_{ji}$ contains odd numbers. Without loss of generality, assume that $\boldsymbol{\beta} = (\beta_1,\ldots,\beta_k,\beta_{k+1},\ldots,\beta_n)$ such that the first $k$ entries are odd and the remaining $n-k$ entries are even.
We construct a SONC polynomial $s_{ji} = a_{\boldsymbol{\alp}(1)} \mathbf{x}^{\boldsymbol{\alp}(1)} + a_{\boldsymbol{\alp}(2)} \mathbf{x}^{\boldsymbol{\alp}(2)} + a_{ji} \mathbf{x}^{\boldsymbol{\beta}}$ such that 
\begin{eqnarray}
 \boldsymbol{\alp}(1) & = & \boldsymbol{\beta} + \sum_{j = 1}^{\lceil k/2 \rceil} \mathbf{e}_j - \sum_{j = \lceil k/2 \rceil + 1}^{k} \mathbf{e}_j, \qquad \boldsymbol{\alp}(2) \ = \ \boldsymbol{\beta} - \sum_{j = 1}^{\lceil k/2 \rceil} \mathbf{e}_j + \sum_{j = \lceil k/2 \rceil + 1}^{k} \mathbf{e}_j, \label{Equ:ProofSONCRepresenation1} \\
 |a_{ji}| & \leq & \sqrt{2 a_{\boldsymbol{\alp}(1)} a_{\boldsymbol{\alp}(2)}}. \label{Equ:ProofSONCRepresenation2}
\end{eqnarray}
By the construction \eqref{Equ:ProofSONCRepresenation1} $\boldsymbol{\alp}(1),\boldsymbol{\alp}(2) \in (2\N)^n$ and $\boldsymbol{\beta} = 1/2 (\boldsymbol{\alp}(1) + \boldsymbol{\alp}(2))$. Thus, $s_{ji}$ is a circuit polynomial and by \eqref{Equ:ProofSONCRepresenation2} the coefficients $a_{\boldsymbol{\alp}(1)}, a_{\boldsymbol{\alp}(2)}$ are chosen large enough such that $|a_{ji}|$ is bound by the circuit number $\sqrt{2 a_{\boldsymbol{\alp}(1)} a_{\boldsymbol{\alp}(2)}}$ corresponding to $s_{ji}$. Thus, $s_{ji}$ is nonnegative by \cite[Theorem 1.1]{Iliman:deWolff:Circuits}. Thus, we obtain 
$$a_{ji} m_{ji} g_j \ = \ s_{ji} g_j + (a_{\boldsymbol{\alp}(1)} \mathbf{x}^{\boldsymbol{\alp}(1)} + a_{\boldsymbol{\alp}(2)} \mathbf{x}^{\boldsymbol{\alp}(2)}) (-g_j),$$
where $s_{ji}$, $a_{\boldsymbol{\alp}(1)} \mathbf{x}^{\boldsymbol{\alp}(1)}$, and $a_{\boldsymbol{\alp}(2)} \mathbf{x}^{\boldsymbol{\alp}(2)}$ are nonnegative circuit polynomials.

\textbf{Degree:} All involved nonnegative circuit polynomials are of degree at most $2d$. In Case 1 this follows by construction. In Case 2 we have for the circuit polynomial $s_{ji}$ that $\deg(\boldsymbol{\alp}(1)),\deg(\boldsymbol{\alp}(2)) = \deg(\boldsymbol{\beta})$ if $k$ is even, and $\deg(\boldsymbol{\alp}(1)) = \deg(\boldsymbol{\beta}) +1$, $\deg(\boldsymbol{\alp}(2)) = \deg(\boldsymbol{\beta})$ if $k$ is odd. Since $\boldsymbol{\beta}$ is an exponent of the polynomial $f$, we know that $\deg(\boldsymbol{\beta}) \leq 2d$. If $k$ is odd, however, then 
$$\deg(\boldsymbol{\beta}) \ = \ \sum_{j = 1}^k \underbrace{\beta_j}_{\text{odd number}} + \sum_{j = k+1}^{n} \underbrace{\beta_j}_{\text{even number}},$$
i.e., $\deg(\boldsymbol{\beta})$ is a sum of $k$ many odd numbers, with $k$ being odd, plus a sum of even numbers. Thus, $\deg(\boldsymbol{\beta})$ has to be an odd number and hence $\deg(\boldsymbol{\beta}) < 2d$. Therefore, all degrees of terms in $s_{ji}$ are bounded by $2d$ and thus $s_{ji} \in C_{n,2d}$.

\textbf{Conclusion:} We have that 
$$f \ = \ \sum_{j = 1}^n p_j g_j \ = \ \sum_{j = 1}^n \sum_{i = 1}^{\ell_j} a_{ji} m_{ji} g_j \ = \ \sum_{j = 1}^n \sum_{i = 1}^{\ell_j} s_{ji} g_j.$$
By Cases 1 and 2 and the degree argument, we have $s_{ji} \in C_{n,2d}$ for every $i,j$ and by defining $s_j = \sum_{i = 1}^{\ell_j} s_{ji} \in C_{n,2d}$ we obtain the desired representation of $f$.
\end{proof}

\subsection{Proof of Theorem~\ref{thm:upper_bound_SONC_hypercube}}
\label{sec:main_proof_constrained}

In this section we combine the results of this section and finish the proof of Theorem~\ref{thm:upper_bound_SONC_hypercube}.

\medskip 

Due to Lemma \ref{Lemma:CertificateIsIndeedCertificate}, it remains to show that $f(\mathbf{x})$ admits a decomposition of the form \eqref{Equ:PolynomialNonnegativityCertificate} if $f(\mathbf{x}) \geq 0$ for every $\mathbf{x} \in \cH_\cP $.

Hence, when restricted to the hypercube $\cH$, the polynomial $f$ can be represented in the following way:
\begin{eqnarray*}
f(\mathbf{x})&=& f(\mathbf{x}) \sum_{\mathbf{v} \in \cH_\mathcal P} \delta_\mathbf{v}(\mathbf{x})  + f(\mathbf{x}) \sum_{\mathbf{v} \in \cH \setminus \cH_\mathcal P} \delta_\mathbf{v}(\mathbf{x})     \qquad \text{ for all } \mathbf{x} \in \cH \;\\
&=&  \sum_{\mathbf{v} \in \cH_\mathcal P} \delta_\mathbf{v}(\mathbf{x}) f(\mathbf{v})  +  \sum_{\mathbf{v} \in \cH \setminus \cH_\mathcal P} \delta_\mathbf{v}(\mathbf{x})f(\mathbf{v})      \qquad \text{ for all } \mathbf{x} \in \cH ,
\end{eqnarray*}
where the last equality follows by Lemma~\ref{lem:delta_v_values} .

Note that there might exist a vector $\mathbf{v} \in \cH \setminus \cH_\mathcal P$ such that $f$ attains a negative value at $\mathbf{v}$. If $f(\mathbf{v}) <0$, then let $p_\mathbf{v} \in \mathcal{P}$ be one of the polynomials among the constraints satisfying $p_\mathbf{v}(\mathbf{v}) <0$. Otherwise, let $p_\mathbf{v}=1$. Since by Lemma~\ref{lem:delta_v_values} we have $\delta_{\mathbf{v}}(\mathbf{x}) p_\mathbf{v}(\mathbf{x})  =   \delta_{\mathbf{v}}(\mathbf{x}) p_\mathbf{v}(\mathbf{v})$ for every $\mathbf{v}, \mathbf{x}  \in \cH$, we can now write:

\begin{eqnarray*}
f(\mathbf{x})&=&  \sum_{\mathbf{v} \in \cH_\mathcal P} \delta_\mathbf{v}(\mathbf{x}) f(\mathbf{v})  +  
\sum_{\mathbf{v} \in \cH \setminus \cH_\mathcal P} \delta_\mathbf{v}(\mathbf{x})p_\mathbf{v}(\mathbf{x}) \frac{f(\mathbf{v})}{p_{\mathbf{v}}(\mathbf{v})}        \qquad \text{ for all } \mathbf{x} \in \cH.
\end{eqnarray*}

Thus, the polynomial $f(\mathbf{x})- \sum_{\mathbf{v} \in \cH_\mathcal P} \delta_\mathbf{v}(\mathbf{x}) f(\mathbf{v}) - 
\sum_{\mathbf{v} \in \cH \setminus \cH_\mathcal P} \delta_\mathbf{v}(\mathbf{x})p_\mathbf{v}(\mathbf{x}) \frac{f(\mathbf{v})}{p_{\mathbf{v}}(\mathbf{v})} $ has degree at most $n+d$ and vanishes on $\cH$. By Theorem~\ref{thm:f_vanishing_on_H_is_SONC_2n} we finally get
\begin{equation*}
f(\mathbf{x}) \ = \ \sum_{j = 1}^n s_j(\mathbf{x}) g_j(\mathbf{x}) + \sum_{j = 1}^n s_{n + j}(\mathbf{x}) (-g_j(\mathbf{x}))   + \sum_{\mathbf{v} \in \cH_\mathcal P} \delta_\mathbf{v}(\mathbf{x}) f(\mathbf{v})  +  
\sum_{\mathbf{v} \in \cH \setminus \cH_\mathcal P} \delta_\mathbf{v}(\mathbf{x})p_\mathbf{v}(\mathbf{x}) \frac{f(\mathbf{v})}{p_{\mathbf{v}}(\mathbf{v})} \;,
\end{equation*}
for some $s_1,\ldots, s_{2n} \in C_{n,n-2}$ and $p_\mathbf{v} \in \mathcal{P}$. This, together with Lemma \ref{lem:Kronecker_in_SONC_2n}, finishes proof together with.
\qed

\begin{cor}
	For every polynomial $f$ which is nonnegative over the boolean hypercube constrained with polynomial inequalities of degree at most $d$ there exists a degree $n+d$ SONC certificate.
\label{Cor:DegreeNplusDCertificate}
\end{cor}
\begin{proof}
	The argument follows directly from Theorem~\ref{thm:upper_bound_SONC_hypercube} by noting that the right hand side of~\eqref{Equ:PolynomialNonnegativityCertificate} is a SONC certificate of degree $n+d$ (see the Definition~\ref{Def:degree_d_certificate}).
\end{proof}

\subsection{Degree $d$ SONC Certificates}
\label{sec:degree_d_certificate}
In this section we show that if a polynomial $f$ admits a degree $d$ SONC certificate, then $f$ also admits a short degree $d$ certificate that involves at most $n^{O(d)}$ terms. We conclude the section with a discussion regarding the time complexity of finding a degree $d$ SONC certificate.

\begin{thm}
	\label{thm:degree_d_certificate}
	Let $f$ be an $n$-variate polynomial, nonnegative on the constrained hypercube $\cH_\cP$ with $|\cP|=\poly(n)$. Assume that there exists a degree $d$ SONC certificate for $f$, then there exists a degree $d$ SONC certificate for $f$ involving at most $O(\binom{n}{\leq d})$ many nonnegative circuit polynomials. 
\end{thm}

\begin{proof}
	Since there exists a degree $d$ SONC proof of the nonnegativity of $f$ on $\cH_\cP$ we know that
	$$
	f(\mathbf{x}) \ = \ \sum_{j}s_j H_j^{(q)},
	$$
	where the summation is finite, the $s_j$'s are SONCs, and every $H_j^{(q)}$ is a product as defined in \eqref{Equ:ProductConstraint}.
	
	\textbf{Step 1:} We analyze the terms $s_j$. Since every $s_j$ is a SONC, we know that there exists a representation
	$$s_j \ = \ \kappa_j \cdot \sum_{i = 1}^{k_j} \mu_{ij} \cdot q_{ij}$$
	such that $\kappa_j, \mu_{1j},\ldots, \mu_{k_jj} \in \R_{> 0}$, $\sum_{i = 1}^{k_j} \mu_{ij} = 1$, and the $q_{ij}$ are nonnegative circuit polynomials. Since $s_j$ is of degree at most $d$, we know that $\struc{Q_j }:= \{q_{1j},\ldots,q_{k_jj}\}$ is contained in $\R_{n,d}[\mathbf{x}]$, which is a real vector space of dimension $\binom{n+d}{d}$. Since $s_j / \kappa_j$ is a convex combination of the $q_{ij}$, i.e. in the convex hull of $Q_j$, and $\dim(Q_j) \leq \binom{n+d}{d}$, applying Carath\'{e}odory's Theorem, see e.g.  \cite{Ziegler}, yields that $s_j/\kappa_j$ can be written as a convex combination of at most $\binom{n+d}{d} + 1$ many of the $c_{ij}$.
	
	\textbf{Step 2:} We analyze the terms $H_j^{(q)}$. By definition of the $\cH_\cP$ and the terms $H_j^{(q)}$ we have
	$$H_j^{(q)} \ = \ g_{j_1} \cdots g_{j_s} \cdot l_{r_1} \cdots l_{r_t} \cdot p_{\ell_1} \cdot p_{\ell_v}$$
	with $j_1,\ldots,j_s \in [n]$, $r_1,\ldots,r_t \in [2n]$, and $\ell_1,\ldots,\ell_v \in [m]$. Since the maximal degree of $H_j^{(q)}$ is $d$, the number of different $H_j^{(q)}$'s is bounded from above by $\binom{n+2n+m}{d}$. 
	
	\textbf{Conclusion:} In summary, we obtain a representation:
	\begin{eqnarray*}
	f(\mathbf{x}) \ = \ \sum_{i=1}^{\binom{n+2n+m}{d}} H_j^{(q)} s_j \ = \ \sum_{i=1}^{\binom{n+2n+m}{d}} H_j^{(q)}  \kappa_j \sum_{j = 1}^{\binom{n+d}{d} +1} \mu_{ij} c_{ij}
	\end{eqnarray*}
	
	Since, as assumed $m$ can bounded by $\poly(n)$, the total number of summands is $\poly(n)^{O(d)}=n^{O(d)}$, and we found a desired representation with at most $n^{O(d)}$ nonnegative circuit polynomials of degree at most $d$.
\end{proof}

The Theorem~\ref{thm:degree_d_certificate} states that when searching for a degree $d$ SONC certificate it is enough to restrict to certificates containing at most $n^{O(d)}$ nonnegative circuit polynomials. Moreover, as proved in~\cite[Theorem 3.2]{Dressler:Iliman:deWolff:Positivstellensatz} for a given set $A \subseteq \mathbb{N}^n$, searching through the space of degree $d$ SONC certificates supported on set $A$ can be computed via a relative entropy program (REP) of size $n^{O(d)}$. 
However, the above arguments do \textit{not} necessarily imply that that the search through the space of degree $d$ SONC certificates can be performed in time $n^{O(d)}$. The difficulty is that one needs to restrict the configuration space of $n$-variate degree $d$ SONCs to a subset of order $n^{O(d)}$ to be able to formulate the corresponding REP in time $n^{O(d)}$. Since the current proof of Theorem \ref{thm:degree_d_certificate} just guarantees the \textit{existence} of a short SONC certificate, it is currently not clear, how to search for a short certificate efficiently. We leave this as an open problem.

\section{There Exists No Equivalent to Putinar's Positivstellensatz for SONCs}
\label{Sec:PutinarPositivstellensatz}
In this section we address the open problem raised in~\cite{Dressler:Iliman:deWolff:Positivstellensatz} asking whether the Theorem~\ref{thm:SONC_Positivstellensatz} can be strengthened by requiring $q=1$. Such a strengthening, for a positive polynomial over some basic closed semialgebraic set, would provide a SONC decomposition equivalent to Putinar's Positivstellensatz for SOS; for background see e.g., \cite{Laurent:Survey,Putinar:Positivstellensatz}. We answer this question in a negative way. More precisely, we provide a polynomial $f$ which is strictly positive over the hypercube $\{\pm 1\}^n$ such that there does not exist a SONC decomposition of $f$ for $q=1$. Moreover, we prove it not only for the most natural choice of the box constraints, that is $l_i =1 \pm x_i$, but for a generic type of box constraints of the form $\ell_i =1+c_i \pm x_i$ with $c_i \in \mathbb{R}_{\geq 0}$. We close the section with a short discussion.

\medskip

Let $\cH=\{\pm 1\}^n$ and consider the following set of polynomials parametrized by a natural number $a$:
$$
\struc{f_a(\mathbf{x})} \ := \ (a-1)\prod_{i=1}^n \left(  \frac{x_i +1}{2}  \right) +1.
$$
These polynomials satisfy $f_a(\mathbf{e}) = a$ for the vector $\mathbf{e}=\sum_{i=1}^n \mathbf{e}_i$ and $f_a(\mathbf{x}) = 1$ for every other $\mathbf{x} \in \cH \setminus \{ \mathbf{e} \}$. We define for every $d \in \N$
$$\struc{S_d} \ := \ \left\{ \sum_{\rm finite} s \cdot h \ : \ s \in  C_{n,2d},~ h \in \left\{1,\pm(x_i^2 -1), 1+c_i \pm x_i \ : \ i \in[n], c_i\in \mathbb{R}_{\geq 0}  \right\}  \right\}$$ 

as the set of polynomials admitting a SONC decomposition over $\cH$ given by Theorem~\ref{thm:SONC_Positivstellensatz} for $q=1$. The main result of this section is the following theorem.
\begin{thm}
	\label{thm:counterexample}
	For every $a > \frac{2^n-1}{2^{n-2}-1}$ we have $f_a \notin S_d$ for all $d \in \N$.
\end{thm}

Before we prove this theorem, we show the following structural results. Note that similar structural observations were already made for AGIforms by Reznick in \cite{Reznick:AGI} using a different notation.
\begin{lemma}
	\label{lem:counterexample_2_values}
	Every $s(\mathbf{x}) \in C_{n,2d}$ attains at most two different values on $\cH =\{\pm 1\}^n$. Moreover, if $s(\mathbf{x})$ attains two different values, then each value is attained for exactly the half of the hypercube vertices. 
\end{lemma}
\begin{proof}
	By Definition~\ref{Def:CircuitPolynomial} every nonnegative circuit polynomial is of the form:
	$$
	s(\mathbf{x}) \ = \ \sum_{j=0}^r f_{\boldsymbol{\alp}(j)} \mathbf{x}^{\boldsymbol{\alp}(j)} + f_{\boldsymbol{\beta}} \mathbf{x}^{\boldsymbol{\beta}}.	
$$
Note that for $j=0,\ldots,r$, we have $\boldsymbol{\alp}(j) \in (2\N)^n$. Hence when evaluated over the hypercube $\mathbf{x} \in \cH =\{\pm 1\}^n$, $s(\mathbf{x})$ can take only one of at most two different values $\sum_{j=0}^rf_{\boldsymbol{\alp}(j)} \pm f_{\boldsymbol{\beta}}$. 

If $s(\mathbf{x})$ attains two different values over $\cH$, then there has to exist a non empty subset of variables that have an odd entry in $\boldsymbol{\beta}$. Let $I\subseteq[n]$ be this subset. Then $s(\mathbf{x})=\sum_{j=0}^rf_{\boldsymbol{\alp}(j)}(\mathbf{x}) - f_{\boldsymbol{\beta}}(\mathbf{x})$, for $\mathbf{x} \in \cH$ if and only if $\mathbf{x}$ has an odd number of $-1$ entries in the set $I$. The number of such vectors is equal to
$$
2^{n-|I|}\sum_{\substack{i=0, \\ i \text{ odd}}}^{|I|} 2^i \ = \ 2^{n-|I|} 2^{|I|-1} \ = \ 2^{n-1}. 
$$
\end{proof}

\begin{lemma}
	\label{lem:counterexample_4_values}
	Every polynomial $s(\mathbf{x})\ell_i(\mathbf{x})$, with $s \in C_{n,2d}$ and $\ell_i = 1 + c_i \pm x_i$ being a box constraint, attains at most four different values on $\cH =\{\pm 1\}^n$. Moreover, each value is attained for at least one forth of the hypercube vertices. 
\end{lemma}
\begin{proof}
	By Lemma~\ref{lem:counterexample_2_values}, $s(\mathbf{x})$ attains at most the two values $\left(\sum_{j=0}^rf_{{\boldsymbol{\alpha}}(j)} \pm f_{\boldsymbol{\beta}}\right)$ on $\cH$. Similarly, $\ell_i(\mathbf{x})$ attains at most the two values $1+c_i \pm x_i $ over $\cH$. Thus, a polynomial $s(\mathbf{x}) \ell_i(\mathbf{x})$ attains at most the four different values $\left(\sum_{j=0}^rf_{{\boldsymbol{\alpha}}(j)} \pm f_{\boldsymbol{\beta}}\right) \left( 1+c_i \pm x_i  \right)$ on $\cH$.
	
	Let $I$ be as in the proof of Lemma~\ref{lem:counterexample_2_values}, i.e., the subset of variables that have an odd entry in $\boldsymbol{\beta}$. If $I=\emptyset$, then the first term $\sum_{j=0}^rf_{{\boldsymbol{\alpha}}(j)} + f_{\boldsymbol{\beta}}$ is constant over the hypercube $\cH$, thus $s(\mathbf{x})\ell_i(\mathbf{x})$ takes two different values depending on the $i$-th entry of the vector. Each value is attained for exactly half of the vectors.
	
	If $I\neq \emptyset$ and $i \notin I$ the claim holds since the value of the first term depends only on the entries in $I$ and the value of the second term depends on the $i$-th entry. Hence, the polynomial $s(\mathbf{x})\ell_i(\mathbf{x})$ attains four values each on exactly one fourth of $\cH$ vectors.
	
	Finally, let $I\neq \emptyset$ and $i \in I$. Partition the hypercube vertices into two sets depending on the $i$-th entry. Each set has cardinality $2^{n-1}$. Consider the set with $x_i=1$. For the vectors in this set the second term takes a constant value $2+c$. Over this set the polynomial $s$ takes one of the values $\sum_{j=0}^rf_{{\boldsymbol{\alpha}}(j)}(\mathbf{x}) \pm f_{\boldsymbol{\beta}}(\mathbf{x})$, depending on whether $\mathbf{x}$ has an odd or even number of $-1$ entries in the set $I\setminus \{-1\}$. In both cases the number of such vectors is equal to
	$$
	2^{n-|I|}\sum_{\substack{i=0, \\ i \text{ odd}}}^{|I|-1} 2^i \ = \ 2^{n-|I|} 2^{|I|-2}=2^{n-2}. 
	$$
	The analysis for the case $x_i=-1$ is analogous.

\end{proof}

Now we can provide the proof of Theorem \ref{thm:counterexample}.

\begin{proof}{(Proof of Theorem~\ref{thm:counterexample})}

Assume $f_a \in S_d$ for some $a \in \mathbb{N}$ and $d \in \N$. We prove that $a$ has to be smaller or equal than $\frac{2^n-1}{2^{n-2}-1}$.

\medskip

Since $f_a \in S_d$ we know that
\begin{eqnarray*}
f_a(\mathbf{x})& = & s_0(\mathbf{x}) + \sum_{i = 1}^n s_i(\mathbf{x}) \ell_i(\mathbf{x}) + \sum_{j = 1}^n \tilde{s}_{j}(\mathbf{x}) (x_j^2 - 1) + \tilde{s}_{j+n}(\mathbf{x}) (1 - x_j^2)
\end{eqnarray*}
with $s_0,\ldots,s_n,\tilde{s}_1,\ldots,\tilde{s}_{2n} \in C_{n,2d}$. Since $\pm (x_j^2-1)$ for $j\in[n]$ vanishes over the hypercube $\cH$, we can conclude
\begin{equation}
\label{eq:counterexample_1}
f_a(\mathbf{x}) \ = \ s_0(\mathbf{x}) + \sum_{i = 1}^n s_i(\mathbf{x}) \ell_i(\mathbf{x}) \qquad \text{ for all } \mathbf{x}\in \cH
\end{equation}
for some $s_0,s_1,\ldots,s_n \in C_{n,2d}$.

Let $s_{0,k}$, and $s_{i,j}$ be some nonnegative circuit polynomials such that $s_0=\sum_k s_{0,k}$, and $s_i= \sum_j s_{i,j}$. Thus, we get
\begin{eqnarray*}
\sum_{\mathbf{x} \in \cH} \left(  s_0(\mathbf{x}) + \sum_{i = 1}^n s_i(\mathbf{x}) \ell_i(\mathbf{x}) \right) & =&
\sum_k \sum_{\mathbf{x} \in \cH}   s_{0,k}(\mathbf{x}) +  \sum_i \sum_j \sum_{\mathbf{x} \in \cH}  s_{i,j}(\mathbf{x}) \ell_{i,j}(\mathbf{x})\\
&\geq & \sum_k 2^{n-1} s_{0,k}(\mathbf{e}) +  \sum_{i = 1}^n \sum_j 2^{n-2} s_{i,j}(\mathbf{e}) \ell_{i,j}(\mathbf{e})\\
&\geq & 2^{n-2} \left(  s_0(\mathbf{e}) + \sum_{i = 1}^n s_i(\mathbf{e}) \ell_i(\mathbf{e}) \right)\\
&=& 2^{n-2} a
\end{eqnarray*}
where the first inequality comes directly from Lemma~\ref{lem:counterexample_2_values} and~\ref{lem:counterexample_4_values} and the last equality from the fact that $f_a(\mathbf{e})=a$.


On the other hand, by the properties of the function $f_a$ and the equality~\eqref{eq:counterexample_1}, we know that
$$
\sum_{\mathbf{x} \in \cH} \left(  s_0(\mathbf{x}) + \sum_{i} s_i(\mathbf{x}) \ell_i(\mathbf{x}) \right) \ = \ 2^{n}-1+ a,
$$
which makes the subsequent inequality a necessary requirement for $f_a \in S_d$: 
$$
a \ \leq \ \frac{2^n-1}{2^{n-2}-1}.
$$

\end{proof}

Speaking from a broader perspective, we interpret Theorem \ref{thm:counterexample} as an indication that the real algebraic structures, which we use to handle sums of squares, do not apply in the same generality to SONCs. We find this not at all surprising from the point of view that in the 19th century Hilbert initially used SOS as a certificate for nonnegativity and many of the algebraic structures in question where developed afterwards with Hilbert's results in mind; see \cite{Reznick:HilbertSurvey} for a historic overview. Our previous work shows that SONCs, in contrast, can, e.g., very well be analyzed with combinatorial methods. We thus see Theorem \ref{thm:counterexample} as further evidence about the very different behavior of SONCs and SOS and as an encouragement to take methods beside the traditional real algebraic ones into account for the successful application of SONCs in the future.

\bibliographystyle{amsalpha}
\bibliography{./SONCIntegerProgramming}
 
\end{document}